\documentclass[twoside]{aiml18}

\usepackage{aiml18macro}

\usepackage[english]{babel}
\usepackage[utf8x]{inputenc}
\usepackage[T1]{fontenc}

\usepackage{amsmath,amssymb,graphicx}
\usepackage{mathtools}
\usepackage{enumerate} 
\usepackage{url}
\usepackage{charter}
\usepackage{graphicx}
\usepackage{verbatim}
\usepackage[colorinlistoftodos]{todonotes}

\usepackage{color}
\usepackage[all]{xy}

\newcommand{\BR}{\tt{BR}}
\newcommand{\FO}{\ensuremath{\mathsf{FO}}}
\newcommand{\FOML}{\ensuremath{\mathsf{FOML}}}
\newcommand{\Bundle}{\mathsf{BFOML}}
\newcommand{\BundleEB}{\mathsf{BFOML^{\exists\Box}}}
\newcommand{\BundleED}{\mathsf{BFOML^{\exists\Diamond}}}
\newcommand{\PSPACE}{\ensuremath{\mathsf{PSPACE}}}
\newcommand{\END}{\tt{END}}
\newcommand{\fixed}{C}
\newcommand{\Var}{\textsf{Var}}
\newcommand{\Sub}{\textsf{Sub}}
\newcommand{\FV}{\textsf{Fv}}

\newcommand{\Ps}{\textbf{P}}

\newcommand{\existsBox}[1]{\exists #1 \Box}
\newcommand{\existsDiamond}[1]{\exists #1 \Diamond}
\newcommand{\forallBox}[1]{\forall #1 \Box}
\newcommand{\forallDiamond}[1]{\forall #1 \Diamond}

\def\lastname{Padmanabha, Ramanujam, Wang}

\begin{document}

\begin{frontmatter}
  \title{Bundled fragments of first-order modal logic: (un)decidability}
  \author{Anantha Padmanabha }
  \author{R. Ramanujam}
  \address{Institute of Mathematical Sciences,
  \\ Homi Bhabha National Institute, Chennai}
   \author{Yanjing Wang}
  \address{Department of Philosophy, Peking University, Beijing}
  
  \begin{abstract}
Quantified modal logic provides a natural logical language for reasoning about modal attitudes even while retaining the richness of quantification for referring to predicates over domains. But then most fragments of the logic are undecidable, over many model classes. Over the years, only a few fragments (such as the monodic) have been shown to be decidable. In this paper, we study fragments that bundle quantifiers and modalities together, inspired by earlier work on epistemic logics of know-how/why/what. As always with quantified modal logics, it makes a significant difference whether the domain stays the same across worlds, or not. In particular, we show that the bundle $\forall \Box$ is undecidable over constant domain interpretations, even with only monadic predicates, whereas $\exists \Box$ bundle is decidable. On the other hand, over increasing domain interpretations, we get decidability with both $\forall \Box$ and $\exists \Box$ bundles with unrestricted predicates. In these cases, we also
obtain tableau based procedures that run in \PSPACE. We further show that the $\exists \Box$ bundle cannot distinguish between constant domain and variable domain interpretations.
  \end{abstract}
 \begin{keyword}
  First-order modal logic, decidability, bundled fragments.
  \end{keyword}
 \end{frontmatter}

\section{Introduction}

In \textit{Meaning and Necessity} \cite{Carnap1947}, Carnap  remarked: 
\begin{quote} \textit{Any system of modal logic without quantification is of interest 
only as a basis for a wider system including quantification. If such a wider system 
were found to be impossible, logicians would probably abandon modal logic entirely.}
\end{quote}
However, it seems that history  went exactly the other way around. Compared to the 
flourishing developments of propositional modal logic in the past decades with successful 
applications in various other fields, first-order modal logic ($\FOML$) is much less studied. 
In addition to numerous philosophical controversies, $\FOML$ is also infamously hard to 
handle technically: e.g., you often loose good properties of first-order logic and modal 
propositional logic when putting them together. 

Among those technical hurdles, finding useful decidable fragments of $\FOML$  has been 
a major one preventing the use of $\FOML$ in computational applications. On the one
hand, the decidable fragments of first-order logic have been well mapped out during 
the last few decades. On the other hand, we have a thorough understanding of the robust 
decidability of propositional modal logics. However, when it comes to finding decidable 
fragments of $\FOML$, the situation seems quite hopeless: even the {\em two-variable fragment 
with one single monadic predicate} is (robustly) undecidable over almost all useful model 
classes \cite{RybakovS17a}. On the positive side, besides the severely restricted one-variable fragment, the 
only promising approach has so far come from the study of the so-called \textit{monodic 
fragment}, which requires that there be at most one free variable in the scope of any 
modal subformula. Combining the monodic restriction with a decidable fragment of $\FO$  we 
often obtain decidable fragments of $\FOML$, as Table \ref{table: FOML} shows (results mostly hold for the usual frame classes e.g., T, S4, $\dots$ ). 

\begin{table}[!htbp]
\label{table: FOML}
\begin{center}
\begin{tabular}{|l|l|l|l|l|}
\hline
Language & Model &  Decidability & Ref\\
\hline 
 $P^1$ & cons-D,  & undecidable &  \cite{Kripke62}\\
 $x,y$, $p$, $P^1$& inc\slash cons-D & undecidable &  \cite{kontchakov2005,gabbay1993}\\
$x,y$, $\Box_i$, single $P^1$ &  inc\slash cons-D  & undecidable & 
\cite{RybakovS17a}
\\
single $x$& inc\slash cons-D & decidable &   \cite{Segerberg73,fischer1978finite}\\
$x,y$\slash$P^1$\slash GF, $\Box_i(x)$ &  inc\slash cons-D  & decidable & \cite{Mono01} \\
\hline
\end{tabular}
\end{center}
\caption{$x,y$ refers to the two-variable fragment, $P^1$ refers to unary predicates. Inc and cons-D refers to increasing domain and constant domain $\FOML$ structures. GF is the guarded fragment. $\Box_i$ is multi-modal logic. $\Box_i (x)$ refers to having only $1$ free variable inside the modality (monodic fragment).
} 
\end{table}

The reason behind this sad tale is not far to seek: the addition of $\Box$ gives implicitly 
an extra quantifier, over a fresh variable. Thus if we consider the two-variable fragment of 
\FOML, with only unary predicates in the syntax, we can use $\Box$ to code up binary relations
and we ride out of the two-variable fragment as well as the monadic fragment of FOL.  The 
monodic restriction confines the use of $\Box$ significantly so that it cannot introduce
a fresh variable implicitly.

It is then natural to ask: apart from variable restrictions, is there some other way to 
obtain syntactic fragments of $\FOML$ that are yet decidable ?

One answer came, perhaps surprisingly, from epistemic logic. In recent years, interest 
has grown in studying epistemic logics of knowing-how, knowing-why, knowing-what, and 
so on (see \cite{WangBKT} for a survey). As observed in \cite{WangBKT}, most of the new 
epistemic operators essentially share a unified \textit{de re} semantic schema of 
$\exists x \Box$ where $\Box$ is an epistemic modality.\footnote{Note that the quantifier 
is not necessarily first-order.} For instance, knowing how to achieve $\phi$ means that 
there exists a mechanism which you know such that executing it will make sure that you end in
a $\phi$ state \cite{Wang17}. Here the distinction between $\exists \Box$ and $\Box \exists$ is crucial. It is also observed that such logics are often decidable. This leads 
to the proposal in \cite{Wang17d} of a new fragment of \FOML by {\sf packing} $\exists$ 
and $\Box$ into a \textit{bundle} modality, but without any restriction on predicates 
or the occurrences of variables. 
$$\phi::= P\overline{x} \mid \neg\phi \mid (\phi\land\phi)\mid \existsBox{x}\phi$$
Note that in this language, quantifiers have to always come with modalities. If $x$ does not appear in $\phi$ then $\exists x\Box$ is simply equivalent to the usual $\Box$. Such a language suffices to say many interesting things besides knowing-how such as:
\begin{itemize}
\item I do not know who killed Mary: $\neg \exists x \Box \textit{kill}(x, \textit{Mary})$
\item I know a theorem such that I do not know any proof: 
$\exists x \Box \neg \exists y \Box \textit{prove}(x,y)$
\item For each person I think it is possible that I know someone who is a friend of him or her: $\forall x \Diamond \exists y \Box \textit{friend}(x,y)$ (note that $\forall x \Diamond $ is $\neg \exists x\Box \neg$)
\item I know a public key whose corresponding private key is known to agent $j$ but not to me: $\exists x \Box_i (\exists y \Box_j \textit{key}(x, y) \land \neg \exists y \Box_i \textit{key}(x, y))$.
\end{itemize}

It is shown that this fragment with arbitrary predicates is in fact \PSPACE-complete 
over increasing domain models. Essentially, the idea is similar to the ``secret of 
success'' of modal logic: \textit{guard the quantifiers}, now with a modality. On the 
other hand, the same fragment is undecidable over S5 models, and this can be shown
by coding first-order sentences in this language using S5 properties.    

There are curious features to observe in this tale of (partial) success. The fragment in 
\cite{Wang17d} includes the $\exists\Box$ bundle but not its companion $\forall \Box$ 
bundle, and considers only increasing domain models. The latter observation is
particularly interesting when we notice that S5 models, where the fragment becomes
undecidable, force constant domain semantics.

The last distinction is familiar to first-order modal logicians, but might come
across as a big fuss to others. Briefly, it is the distinction between a {\em possibilist}
approach and an {\em actualist} approach. In the former, the model has one fixed
domain for all possible worlds, and quantification extends over the domain (rather
than only those objects that exist in the current world). This is the constant domain
semantics. In the latter approach, each possible world has its own domain, and
quantification extends only over objects that exist in the current world. In increasing
domain semantics, once an object exists in a world $w$, it exists in worlds accessible
from $w$.

Given such subtlety, it is instructive to consider more general bundled fragments
of $\FOML$, including both $\exists\Box$ and $\forall \Box$ as the natural first step, and study them over constant
domain as well as increasing domain models. This is precisely the project undertaken
in this paper, and the results are summarized in Table 2.

\begin{table}[!ht]
\label{tab:bundle}
\begin{center}
\begin{tabular}{|l|l|l|l|}
\hline
Language & Domain& Decidability &  Remark\\
\hline 
$\forall\Box$, $P^1$ & constant& undecidable & \\
\hline
$\exists\Box$, $P$  &constant& decidable & \PSPACE-complete\\
\hline
$\exists\Box, \forall\Box,\ P$& increasing& decidable&  \PSPACE-complete \\ 
\hline
\end{tabular}
\end{center}
\caption{Satisfiability problem classification for Bundled $\FOML$ fragment}
\end{table}

As we can see, the $\exists \Box$ bundle behaves better computationally than the 
$\forall \Box$ bundle. For $\forall\Box$, even the monadic fragment is undecidable over 
constant domain models: we can encode in this language, qua satisfiability, any first-order 
logic sentence with binary predicates by exploiting the power of $\forall \Box$. A 
straightforward consequence is that the $\Box\forall$ fragment is also undecidable 
over constant domain models.

On the other hand, we can actually give a tableau method for the $\exists\Box$ and $\forall\Box$ fragment together, 
similarly as the tableau in \cite{Wang17d}, for increasing domain models. The crucial 
observation is that such models allow us to manufacture new witnesses for 
$\forall x \Diamond$ and $\exists x \Diamond$ formulas, giving considerable
freedom in model construction, which is not available in constant domain models.

Indeed, the well-behavedness of the $\exists\Box$ bundle is further attested to
by the fact that it is decidable over constant domain models as well. So constant
domain is not the culprit for undecidability of this fragment over S5 models. Actually, we can show the $\exists \Box$ bundle does not distinguish increasing domain models and constant domain models. 

The paper is structured as follows. After formal definitions of bundled fragments,
we present undecidability results and then move on to tableaux procedures for the 
decidable fragments. We then show that the validities of $\exists\Box$ over 
increasing domain are exactly the same as its validities over constant domain 
models, and end the paper with a re-look at mapping the terrain of these fragments.

\section{The bundled fragment of $\FOML$}
\label{sec: syntax and semantics}
Let $\Var$ be a countable set of variables, and $\Ps$ be a fixed set of predicate symbols,
with  $\Ps^n \subseteq \Ps$ denoting the set of all predicate symbols of arity $n$.  We 
use $\overline{x}$ to denote a finite sequence of (distinct) variables in $\Var$.  We only
consider the `pure' first order unimodal logic: that is, the vocabulary is restricted to 
$\Var$ (no equality and no constants and no function symbols).

\begin{definition}
\label{def: bundle syntax}
Given  $\Var$ and $\Ps$,  the bundled fragment of $\FOML$ denoted by $\Bundle$ is defined 
as follows:
$$ \phi::= P\overline{x} \mid \neg\phi \mid (\phi\land\phi)\mid \existsBox{x}\phi  \mid \existsDiamond{x} \phi$$
where $x \in \Var$, $P\in \Ps$. We denote the fragment $\BundleEB$ to be the formulas which 
contains only $\existsBox{ }$ formulas and $\BundleED$ which contains only $\existsDiamond{ }$ 
formulas. 
\end{definition}

$\top, \bot, \lor, \Rightarrow$ (True, False, Or and Implies) are defined in the 
standard way. $\forallDiamond {x} \phi$ as $\neg \existsBox {x} \neg\phi$ is the 
dual of $\existsBox{x}\phi$, and $\forallBox{x}\phi$ defined by $\neg \existsDiamond{x} \neg\phi$ 
is the dual of $\existsBox{x}\phi$. With both bundles we can say, in an epistemic setting, that for each country I know its capital $\forall x\Box \exists y \Box \textit{Capital}(x,y).$ 

The \textit{free} and \textit{bound} occurrences of variables are defined as in first-order logic, 
by viewing $\existsBox{x}$ and $\existsDiamond{x}$ as quantifiers. We denote $\FV(\phi)$ as the set of free variables of $\phi$. We write $\phi(\overline{x})$ 
if all the free variables in $\phi$ are included in $\overline{x}$. Given a $\Bundle$ formula 
$\phi$ and $x, y\in \Var$, we write $\phi[y\slash x]$ for the formula obtained by replacing 
every free occurrence of $x$ by $y$. A formula is said to be a {\em sentence} if it contains
no free variables.

The semantics is the standard \textit{increasing domain} semantics of $\FOML$.
The $\vDash$ relation is specialized to the $\Bundle$ fragment.
 
\begin{definition}
\label{def: bundle sematnics}
An (increasing domain) model $M$ for $\Bundle$ is a tuple $(W, D, \delta, R, \rho)$ where, 
$W$ is a non-empty set of worlds, $D$ is a non-empty domain, $R\subseteq (W\times W)$,
$\delta:W\to 2^D$ assigns to each $w\in W$ a \textit{non-empty} local domain s.t. 
$wRv$ implies $\delta(w)\subseteq \delta(v)$ for any $w,v\in W$, and
$\rho: W\times \Ps \to \bigcup_{n\in \omega}2^{D^n}$ such that $\rho$ assigns to
each $n$-ary predicate on each world an $n$-ary relation on $D$.
\end{definition}

We often write $D_w$ for $\delta(w)$.  A \textit{constant domain} model is one where
$D_w=D$ for any $w\in W$.  A \textit{finite model} is one with both $W$ finite and $D$ 
finite. 

Consider a model $M = (W, D, \delta, R, \rho)$, $w \in W$. To interpret free 
variables, we also need a variable assignment $\sigma: \Var\to D$.  Call $\sigma$
{\em relevant} at $w \in W$ if $\sigma(x)\in \delta^M(w)$ for all $x\in \Var$. The
increasing domain condition ensures that whenever $\sigma$ is relevant at $w$ and
we have $wRv$, then $\sigma$ is relevant at $v$ as well. (In a constant domain
model, every assignment $\sigma$ is relevant at all the worlds.)

Given $M = (W, D, \delta, R, \rho)$, $w \in W$, and an assignment $\sigma$ 
relevant at $w$, define $M,w,\sigma \vDash \phi$ inductively as follows:
$$\begin{array}{|lcl|}
\hline
M, w, \sigma\vDash P(x_1\cdots x_n) &\Leftrightarrow & (\sigma(x_1), \cdots, \sigma(x_n))\in \rho(P,w)  \\ 
M, w, \sigma\vDash \neg\phi &\Leftrightarrow&   M, w, \sigma\nvDash \phi \\ 
M, w, \sigma\vDash (\phi\land \psi) &\Leftrightarrow&  M, w, \sigma\vDash \phi \text{ and } M, w, \sigma\vDash \psi \\ 
M, w, \sigma\vDash \existsBox{x} \phi &\Leftrightarrow& \text{there is some $d\in \delta(w)$ such that }\\
&&M, v, \sigma[x\mapsto d]\vDash\phi \text{ for all $v$ s.t.\ $wRv$}\\
M, w, \sigma\vDash \existsDiamond{x} \phi &\Leftrightarrow& \text{there is some $d\in \delta(w)$ and some $v\in W$ }\\
&&\text{ such that } wRv \text{ and } M, v, \sigma[x\mapsto d]\vDash\phi \\
\hline
\end{array}$$

\noindent where $\sigma[x\mapsto d]$ denotes another assignment that is the same as $\sigma$ 
except for mapping $x$ to $d$. 

It is easily verified that $M,w,\sigma \vDash \phi$ is defined only when $\sigma$ is 
relevant at $w$. In general, when considering the truth of $\phi$ in a model, it
suffices to consider $\sigma: \FV(\phi) \to D$, assignment restricted to the free variables
occurring free in $\phi$.  When $\FV(\phi) = \{x_1, \ldots, x_n\}$ and $\{d_1, \ldots, d_n\}
\subseteq D$, We  write $M, w\vDash \phi [\overline{d}]$ to denote $M,w,\sigma\vDash 
\phi(\overline{x})$ for any $\sigma$ such that and for all $i\le n$ we have $\sigma(x_i) = d_i$. 
Hence when $\phi$ is a sentence, we can simply write $M,w \models \phi$.

We say $\phi$ is \textit{valid}, if $\phi$ is true on any $M, w$ w.r.t. any $\sigma$ relevant
at $w$.  $\phi$ is \textit{satisfiable} if $\neg \phi$ is not valid.

\section{Undecidability results}

In this section we prove that the satisfiability problem for the $\BundleED$ fragment 
with constant domain semantics is undecidable even when the atomic predicates are 
restricted to be unary. We prove this by reduction from the  satisfiability problem for 
$\FO$ with one arbitrary binary predicate, which is known to be undecidable (from \cite{godel1933}).

That full $\FOML$ with constant domain semantics is undecidable even when the atomic 
predicates are only unary is well known; it was shown by  Kripke\cite{Kripke62}. That we need
only 2 variables along with propositions to make Monadic $\FOML$ undecidable was shown
by Gabbay and Shehtman \cite{gabbay1993}. That propositions can be eliminated was observed
by Kontchakov, Kurucz and Zakharyaschev  \cite{kontchakov2005}. 

Consider $\FO(R)$, the first order logic with only variables as terms and no equality,
and the single binary predicate $R$. To translate $\FO(R)$ sentences to $\BundleED$ formulas,
we use two unary predicate symbols $p,q$ in the latter. The main idea is that the atomic
formula $R(x,y)$ is coded up as the $\FOML$ formula $\exists z \Diamond \big( p(x) \land q(y)\big)$,
where $z$ is a new variable, distinct from $x$ and $y$.\footnote{This is similar to the approach used by Kripke \cite{Kripke62}, specialized for the $\BundleED$ fragment.} In the model constructed, it will
turn out that $R = P \times Q$. But in which world is this to be enforced ? We will enforce
that all worlds at a specific modal depth interpret $R$ in the same way, thus ruling out
any ambiguity, crucially using the $\forall \Box$ bundle and constant domain semantics.

For any {\em quantifier free} $\FO(R)$ formula $\alpha$, we define the translation 
of $\alpha$ to $\BundleED$ formula $\phi_\alpha$ inductively as follows. 

\begin{itemize}
\item[-] $\phi_{R(x,y)} ::= \exists z \Diamond \big( p(x) \land q(y)\big)$, where $z$ is
distinct from $x$ and $y$. 
\item[-] $\phi_{\neg \alpha} ::= \neg \phi_{\alpha}$ and
$\phi_{\alpha_1 \land \alpha_2} ::= \phi_{\alpha_1} \land \phi_{\alpha_2}$.
\end{itemize}

Now consider an $\FO(R)$ sentence $\alpha$ (having no free variables) and presented in 
prenex form: $Q_1 x_1\ Q_2 x_2 \cdots Q_n x_n (\beta)$ where $\beta$ is quantifier free.
We define $\psi_\alpha$ to be the conjunction of the following three sentences:

\begin{itemize}
\item[-] $\psi_1 ::= Q_1 x_1 \Delta_1\  Q_2 x_2 \Delta_2\ \cdots Q_n x_n\Delta_n\  (\phi_{\beta})$ \\
where $Q_i x_i \Delta_i :=  \existsDiamond{x_i}$ if $Q_i = \exists$ and
$Q_i x_i \Delta_i := \forallBox{x_i}$ if $Q_i = \forall$.

\item[-] $\psi_2 ::= \forallBox{z_1}\forallBox{z_2}\big( (\existsDiamond{z})^{n}(\exists z\Diamond (p(z_1) \land q(z_2))) \Rightarrow (\forallBox{z})^n(\exists z\Diamond (p(z_1) \land q(z_2))) \big)$.

\item[-] $\psi_3 ::= \bigwedge_{j=1}^{n+2}(\forall\Box{z})^{j}\exists z\Diamond \top$.
\end{itemize}

Of these, $\psi_1$ ensures that the formulas are interpreted over the same domain, and that
the meaning of $R$ is given as $P \times Q$ in a world at depth $n+1$. $\psi_2$ ensures that
all worlds at depth $n+3$ agree on $p$ and $q$ and hence on $R$. $\psi_3$ asserts that every
path can be extended until depth $n+2$, one never gets stuck earlier.

The role of dummy variables $z_1$ and $z_2$ in $\psi_2$ and depth $n+2$ in $\psi_3$
may need an explanation. First note that the interpretation for $R$ is collected at 
depth $n+1$, and the extra depth is because the extra quantification in the coding of 
$R$ added successor worlds. Now we need these two variables to refer to elements of 
$p$ and $q$ at depth $n+1$, but in the bundled fragment, any variable comes packed 
with a modality. Thus we get depth $n+3$. Further we could not use variables from
$\alpha$ (which might be quantified existentially), so fresh variables are needed.
  
\begin{theorem}
\label{thm: bundle-ED undecidability on constant domain}
The $\FO(R)$ sentence $\alpha$ is satisfiable iff the $\BundleED$ sentence $\psi_\alpha$ is
constant domain satisfiable.
\end{theorem}

\begin{proof}
We sketch the proof here, the details are given in Appendix A.
Fix $\alpha ::= Q_1 x_1 \cdots Q_n x_n \beta$, where $\beta$ is quantifier free. 
To prove $(\Rightarrow)$, assume that $\alpha$ is satisfiable. Let $D$ be some
domain such that $(D,I) \models \alpha$ where $I \subseteq (D \times D)$ is the 
interpretation for $R$.

Define $M = (W,R,D,\delta,\rho)$ where: 
\begin{itemize}
\item[ ]$W = \{ v_1, v_2 \} \cup \{w_i \mid 1 \le i \le n\} \cup \{ u_d \mid d \in D\}$.
\item[ ] $R =  \{ (v_1,v_2), (v_2,w_1)\} \cup \{ (w_i, w_{i+1}) \mid 1 \le i < n\} \cup \{ (w_n,u_d) \mid u_d \in W\}$.
\item[ ] $\delta(u) = D$ for all $u \in W$.
\item[ ] For all $i \in \{ 1,2\}$ and $1\le j \le n$ and $v_i,w_j \in W$ define $\rho(v_i,p) = \rho(v_i,q) =  \rho(w_j, p) = \rho(w_j,q) = \emptyset$ and for all $u_d \in W,\ \rho(u_d,p) = \{ d\}$ and $\rho	(u_d, q) = \{ c \mid (d,c) \in I \}$.
\end{itemize}

By construction, $M$ is a model that is a path of length $n+2$ originating from
$v_1$ until $w_n$ at which point we have a tree of depth $1$, with children $u_d$,
one for each $d \in D$. Therefore, it is easy to see that $M, v_1 \models \psi_3$.

Note that $M$ is a constant domain model. Further, it can be easily checked that
$(a,b) \in R$ iff $M, u_a \models (p(a) \land q(b))$. Thus $(D,I) \models R(x,y)$ iff
$M, w_n \models \exists z\Diamond (p(x) \land q(y))$. Hence a routine induction shows
that for any quantifier free formula $\beta'$, $(D,I) \models \beta'$ iff $M, w_n \models 
\phi_{\beta'}$. Further, since $M$ is a path model until $w_n$ and there is a path of length $n+3$ starting from $v_1$, we see that
$M, v_1 \models \psi_2 \land \psi_3$. We then show that $M, v_1 \models \psi_1$, which would
complete the forward direction of the proof. This is proved by reverse induction on $i$. The 
base case, when $i = n$, follows from our assertion above on the interpretation of $R$ at $w_n$.
For the induction step, we crucially use the fact the model constructed is a path and hence
$\Box$ and $\Diamond$ coincide along the path.

\paragraph{}
To prove $(\Leftarrow)$, suppose that $\psi_\alpha$ is satisfiable, and let  
$M= (W,D,R,\gamma,V)$ be a constant domain model such that $M,v \models \psi_\alpha$.
Without loss of generality, we can assume $(W,R)$ to be a tree rooted at $v$, and
$\psi_3$ ensures that every path in it has length at least $n+3$.  
Let $u'$ be any world at height $n+3$. Define $I_{u'} = \{ (c,d) \mid c \in \rho(u',p),
d \in \rho(u',q)\}$. For world $u$ at height $n+2$, define $I_u = \bigcup \{I_{u'} \mid
(u,u') \in R\}$. Since $M,v \models \psi_2$, we see that $I_u = I_w$, for all $u, w$
at height $n+2$. Hence we unambiguously define $I = I_u$, thus defining the first
order model $M' = (D,I)$. We now claim that the formula $\alpha$ is satisfied in
this model. The definition of $I$ ensures that the atomic formulas are correctly
satisfied. We proceed by subtree induction noting that all children of a node
satisfy subformulas equivalently (which is needed for $\forall\Box$ formulas).

\end{proof}

\section{Decidability results}
Having seen that the $\Bundle$ fragment is undecidable over constant domain models,
and noted that the $\exists \Box$ bundle is decidable over increasing domain models
(\cite{Wang17d}), it is natural to wonder whether the problem is decidable with the $\forall\Box$
bundle or constant domain semantics, or both. In this section, we show that it is
indeed the combination that is the culprit, by showing that relaxing either of the
conditions leads to decidability. First, we show that the full fragment is decidable
over increasing domain models, and then show that the $\exists \Box$ bundle is 
decidable over constant domain models.

\subsection{Increasing domain models}

We consider formulas given in negation
normal form (NNF): 

$$ \phi::=  P\overline{x} \mid \neg P\overline{x} \mid (\phi\land \phi) \mid (\phi\lor\phi)\mid \existsBox{x} \phi \mid \existsDiamond{x} \phi \mid \forallBox{x} \phi \mid \forallDiamond{x} \phi $$

Formulas of the form $P\overline{x}$ and $\neg P \overline{x}$ are 
\emph{literals}. Clearly, every $\Bundle$-formula $\phi$ can be rewritten into 
an equivalent formula in NNF. 

We call a formula {\em clean} if no variable occurs both bound and free in it and
every use of a quantifier quantifies a distinct variable. Note that every 
$\Bundle$-formula can be rewritten into an equivalent clean formula.  (For instance,
$\exists x \Box P(x) \lor \exists x \Box Q(x)$ and $P(x) \land \exists x \Box Q(x)$ 
are unclean formulas, whereas $\exists x \Box P(x) \lor \exists y \Box Q(y)$ and 
$P(x) \land \exists y \Box Q(y)$ are their clean equivalents.)

We define the following tableau rules for all $\Bundle$ formulas in NNF. The tableau
is a tree structure $T = (W,V,E,\lambda)$ where $W$ is a finite set, $(V,E)$ is a 
rooted tree and $\lambda: V \to L$ is a labelling map. Each element in $L$ is of 
the form $(w,\Gamma,F)$, where $w \in W$, $\Gamma$ is a finite 
set of formulas and $F \subseteq \Var$ is a finite set.  The intended meaning of 
the label is that the node constitutes a world $w$ that satisfies the formulas in 
$\Gamma$ with the `assignment' $F$, with each variable in $F$ denoting one that 
occurs free in $\Gamma$ and as we will see, the interpretation will be the identity.

 A rule specifies that if a node labelled
by the premise of the rule exists, it can cause one or more new nodes to be created as 
children with the labels as given by the completion of the rule. 

\begin{definition}{Tableau rules}

{\scriptsize 
\begin{center}
\noindent \begin{tabular}{|c|}
\hline
$\dfrac{w: \phi_1\lor\phi_2,\Gamma, F}{w:\phi_1,\Gamma, \sigma \mid w:\phi_2,\Gamma, F}$ \tt{($\lor$)}\qquad $\dfrac{w: \phi_1\land\phi_2,\Gamma, F}{w:\phi_1,\phi_{2},\Gamma,F}${\tt ($\land$)}\\
\hline
Given $n_1,m_1\geq 1, n_2,m_2,s \geq 0$:\\
$\dfrac{\splitdfrac{w:\existsDiamond{x_1}\alpha_1 \cdots,\existsDiamond{x_{n_1}}\alpha_{n_1}, \existsBox{y_1}\beta_1,\cdots, \existsBox{y_{n_2}}\beta_{n_2},
}{\forallDiamond{z_1}\phi_1,\cdots, \forallDiamond{z_{m_1}}\phi_{m_1},\forallBox{z'_1}\psi_1,\cdots,\forallBox{z'_{m_2}}\psi_{m_2},}{r_1\dots r_s, F}}
{\splitdfrac { \langle wv_{x_i}: \alpha_i, \{ \beta_j \mid 1 \le j \le n_2\} , \{\psi_l[z/z'_l] \mid z \in F', l\in [1,m_2]\} , F' \rangle \cup}{\langle wv^y_{z_k}: \phi_k[y/z_k], \{ \beta_j \mid 1 \le j \le n_2\}, \{\psi_l[z/z'_l] \mid z \in F', l\in [1,m_2]\} , F' \rangle}}$ \tt{($\BR$)}\\
${where\ i \in [1,n_1], k\in [1,m_1], y\in F'}$  \\
\hline
Given $n_2, m_2\geq 1;\ s\geq 0$:\\
$\dfrac{w: {\existsBox{y_1} \beta_1, \cdots, \existsBox{ y_{n_2}} \beta_{n_2},\forallBox{z'_1}\psi_1,\cdots,\forallBox{z'_{m_2}}\psi_{m_2}, r_1\dots r_s, F}}{w: r_1\dots r_s, F}  $ (\END) \\
\hline
\end{tabular}\\

\noindent where $F'= F \cup\{x_i\mid i\in [1,n_1] \ \cup \{y_j \mid j \in [1,n_2]\}$ and $r_1\cdots r_s\in lit$ (the literals).
\end{center}
}
\end{definition}

The rules are standard. The rule $(\END)$ says that in the absence of any $Q x \Diamond$
formulas, with $Q \in \{\exists, \forall\}$, the branch does not need to be explored further,
only the literals remain. Further, note that there is an implicit ordering on how rules
are applied: $(\BR)$ insists on the label containing no top level conjuncts or disjuncts,
and hence may be applied only after the $\land$ and $\lor$ rules have been applied as 
many times as necessary.

The rule $(\BR)$ looks complicated but asserts standard modal
validities, but with multiplicity. To see how it works, consider a model $M$, a world
$u$ and assignment $\sigma$ such that $(M, u, \sigma) \models \exists x \Diamond \alpha 
\land \exists y \Box \beta \land \forall z \Box \psi$. Then for some domain element
$c \in \delta(u)$, we have a successor world $v$ such that $(M,v,\sigma') \models
\alpha \land \beta \land \psi$, where $\sigma'(x) = c$ and $\sigma'(z) = c$. Further
if $(M, u, \sigma) \models \forall z \Diamond \phi \land \forall z' \Box \psi$ then
for every domain element $d \in \delta(u)$, we have a successor world $v_d$ such that
$(M,v_d,\sigma') \models \phi \land \psi$, where $\sigma'(z) = d$ and $\sigma'(z') = d$.
When the domain elements we use are themselves variables, they can be substituted 
into formulas so we could well write $(M,v_d,\sigma') \models \phi[z/d] \land \psi[z'/d]$.
The rule $(\BR)$ achieves just this, but has to do all this simultaneously for all the
quantified formulas at the node ``in one shot'', and has to keep the formulas clean too.

We need to check that the rule $(\BR)$ is well-defined. Specifically, if the label in
the premise contains only clean formulas, we need to check that the label in the
conclusion does the same. To see this, observe the following, with $\Gamma$ being
the set of clean formulas in the premise. Let $\Delta, \Delta'$ stand for any modality.

\begin{itemize}
\item Note that if $\exists x \Delta \phi$ and $Q y \Delta' \psi$ are both
in $\Gamma$, with $Q$ any quantifier, then $x \neq y$ and neither $x$ occurs free in $\psi$ nor $y$ 
occurs free in $\phi$, also $\phi$ or $\psi$ do not contain any subformula
that quantifies over $x$ or $y$.

\item Hence, in the conclusion of $(\BR)$, every substitution of the form 
$\phi[y/z]$ or $\psi[z/z']$ results in a clean formula, since $y$ occurs free
in $\phi$ and $z$ does not occur at all in $\phi$ and similarly for $\psi$.
\end{itemize}

Thus, maintaining `cleanliness' allows us to treat existential quantifiers as
giving their own witnesses. The `increase' in the domain is given by the added
elements in $F'$ in the conclusion.  
Note that with each node creation either the number of boolean connectives or the
maximum quantifier rank of formulas in the label goes down, and hence repeated applications
of the tableau rules must terminate, thus guaranteeing that the tableau generated is
always finite. 

A tableau is said to be \textit{open} if it does not contain any node 
$u$ such that its label contains a literal $r$ as well as its negation.  
Given a tableau $T$, we say a node $(w: \Gamma, F)$ is a \textit{branching node} 
if it is branching due to the application of $\BR$.  We call $(w, \Gamma, F)$ the 
\textit{last node of} $w$, if it is a leaf node or a branching node. Clearly, given any 
label $w$ appearing in any node of a tableau $T$, the last node of $w$ uniquely exists. 
If it is a non-leaf node, every child of $w$ is labelled $wu$ for some $u$.

Let $t_w$ denote the last node of $w$ in tableau $T$ and let
$\lambda(t_w) = (w,\Gamma,F)$. If it is a non-leaf node, then it is a branching
node with rule $(\BR)$ applying to it with $F'$ as its conclusion.
We let $Dom(t_w)$ denote the set $F'$ in this case and $Dom(t_w) = F$ otherwise.

\begin{theorem}
\label{thm: Bundle increasing decidability}
For any clean $\Bundle$-formula $\theta$ in NNF, 
there is an open tableau from $(r, \{\theta\}, F_r)$ where 
$F_r =\{x \mid x \text{ is free in } \theta \}\cup 
\{z\}$, where  $z\in \Var,\ z$ does not appear in $\theta$, iff
$\theta$ is satisfiable in an increasing domain model.  
\end{theorem}

\begin{proof}

Let $T$ be any tableau
$T$ starting from $(r, \{\theta\}, F_r)$ where $\theta$ is clean. We observe
that for any node $(v, \Gamma, F)$ in $T$, we have the following.
If $x \in F$ and occurs in a formula in $\Gamma$ then every occurrence is free. Further,
every variable $x$ occurring free in a formula in $\Gamma$ is in $F$. 
These are proved by induction on the structure of $T$ using the fact that the
rule $(\BR)$, when applied to clean formulas, results in clean formulas.

To prove the theorem, given an open tableau $T = (W,V,E,\lambda)$ with root node labelled by
$(r,\{ \theta\}, F_r)$, we define $M=(W, D, \delta, R, \rho)$ where: $D=\Var$; $w R v$ iff 
$v=wv'$ for some $v'$; $\delta(w)=Dom(t_w)$; $\overline{x}\in \rho(w, P)$ iff $P\overline{x} 
\in \Gamma$, where $\lambda(t_w) = (w, \Gamma, F)$.
Clearly, if $w R v$ then $Dom(t_w) \subseteq Dom(t_v)$, and hence
$M$ is indeed an increasing domain model.  

Moreover $\rho$ is well-defined due to openness of $T$. 
We now show that $M,r$ is indeed a model of $\theta$, and this is proved by the
following claim.

\paragraph{Claim.} For any tree node $w$ in $T$ if $\lambda(t_w) =
(w, \Gamma, F)$ and if $\alpha \in \Gamma$ then $(M,w,id_F) \models \alpha$.
(Below, we abuse notation and write $(M,w,F) \models \alpha$
for $(M,w,id_F) \models \alpha$ where $id_F = \{ (x,x) \mid x \in F\}$.)

The proof proceeds by subtree induction on the structure of $T$. The base case
is when the node considered is a leaf node and hence it is also the last node
with that label. The definition of $\rho$ ensures that the literals are
evaluated correctly in the model.

For the induction step, the cases for the conjunction and disjunction rules
are standard. Now consider the application of rule $(\BR)$ at a branching
node $t_w$ with label $(w, \Gamma, F)$. Let
$$\begin{array}{ll}
\Gamma=&
\{\existsDiamond{x_i}\alpha_i \mid i\in [1,n_1]\}  \cup \{\existsBox{y_j}\beta_j \mid j\in [1,n_2]\} \cup \{\forallDiamond{z_k}\phi_k \mid k \in [1,m_1]\}\\
&\cup
\{\forallBox{z'_l}\psi_l \mid l \in [1,m_2] \}\cup \{r_1\dots r_s\}.
\end{array}$$
By induction hypothesis, we have that for every $i \le n_1$,
$M,wv_{x_i}, F' \vDash \alpha_i \land \bigwedge_{j \le n_2} \beta_j \land \psi'$
and for every $y\in F'$ and $k \in [1, m_1]$, 
$M,wv^y_{z_k}, F' \vDash \phi_k[y \slash z_k] \land \psi'$, where
$\psi' = \bigwedge_{ l\le m_2}^{z \in F'}\psi_{l}[z\slash z_l']$.

Note that $D_w=Dom(t_w)=F'$. We need to show that $M, w, F \vDash \alpha$ for each 
$\alpha \in \Gamma$. Every such $\alpha$ is either a literal or a bundle formula.
The assertion for literals follows from the definition of $\rho$.
For $\exists x_i \Diamond \alpha_i \in \Gamma$ we have the successor $wv_{x_i}$ 
where $\alpha_i$ is true. Similarly for every $\forall z_k \Diamond \phi_k \in \Gamma$ 
and $y \in D_w$ we have the successor $wv^{y}_{z_k}$ where $\phi_l[y \slash z_k]$ is true. 

Now for the case $\exists y_j \Box  \beta_j$: by induction hypothesis, for all 
successors $wv^{\#}_z$ of $w$ where $\#$ is either empty or $\# \in F'$ 
we have $M, wv^{\#}_{z},F'\vDash \beta_j$. By cleanliness of $\beta_j$, for all
$j' \ne j$ and for all $i,k,l$ we have that $x_i, y_{j'},z_k,z'_l$ are not free 
in $\beta_j$. Hence $M, wv^{\#}_{z}, id_F[y_j\mapsto y_j]\vDash \beta_j$ for 
each $wv^{\#}_z$.  Since $y_j \in F' = D_w$ we have $M,w, id_F \vDash \exists y_j \Box \beta_j$. 

The case $\forall z_l' \Box \psi_l$ is similar. By induction hypothesis, 
we have $M, wv^{\#}_{z}, F'\vDash \psi_l[a \slash z_l']$ for every $a\in F'$ 
and again by cleanliness of $\psi_l$, for all $l' \ne l$ and for all $i,j,k$ 
we note that $x_i,y_j,z_k,z'_{l'}$ are not free in $\psi_l$. Thus 
$M, wv^{\#}_{z}, F'[z_l \mapsto a] \vDash \psi_l$ for all $a \in F' = D_w$.
Hence $M,w, id_F \vDash \forallBox{z'_l} \psi_l$.

Thus the claim is proved and hence it follows that $M,r,F_r\vDash \theta$. 

\medskip

\noindent \textbf{Completeness of tableau construction}: \\
We only need to show that all rule applications preserve the satisfiability of 
the formula sets in the labels. This would ensure that there is an open tableau 
since satisfiability of formula sets ensures lack of contradiction among literals.
It is easy to see that the rules $(\land)$ and $(\END)$ preserve satisfiability.
If one of the conclusions of the $(\lor)$ rule is satisfiable then so is the premise.
It remains only to show that $(\BR)$ preserves satisfiability. Consider a label set
$\Gamma$ of clean formulas at a branching node. Let
$\begin{array}{ll}
\Gamma=&
\{\existsDiamond{x_i}\alpha_i \mid i\in [1,n_1]\}  \cup \{\existsBox{y_j}\beta_j \mid j\in [1,n_2]\} \cup \{\forallDiamond{z_k}\phi_k \mid k \in [1,m_1]\}\\
&\cup
\{\forallBox{z'_l}\psi_l \mid l \in [1,m_2] \}\cup \{r_1\dots r_s\}.
\end{array}$
be satisfiable at a model $M=\{W, D, \delta, R, \rho\}$, $w \in W$ and an 
assignment $\eta$ such that $\eta(x)\in D_w$ 
for all $x\in FV(\Gamma)$ and $M, w, \eta \vDash \bigwedge_{\chi \in \Gamma} 
\chi $.

By the semantics, we have the following: 
(A):  There exist $a_1, \dots, a_{n_1}\in D_w$ and $v_1 \dots v_{n_1} \in W$ 
successors of $w$ such that $M,v_i, \eta[x_i \mapsto a_i] \vDash \alpha_i$.
(B): There exist $b_1, \dots b_{n_2} \in D_w$ such that for all $v \in W$ 
if $w \to v$ then $M,v, \eta[y_j \mapsto b_j] \models \beta_j$.
(C): For all $c\in D_w$ there exist $v_1^c \dots v^c_{m_1} \in W$, successors 
of $w$ such that $M,v^c_k,\eta[z_k \mapsto c] \models \phi_k$.
(D): For all $d\in D_w$ and for all $v \in W$ if $w \to v$ then 
$M,v,\eta[z'_l \mapsto d] \models \psi_l$.

By cleanliness of $\beta_j$, each $y_j$ is free only in $\beta_j$ and $y_j$ 
is not free in any $\alpha_i, \beta_{j'}, \phi_k, \psi_l$ for $j'\ne j$ and 
for all $i,k,l$. Similarly $z'_l$ is free only in $\psi_l$ and $z'_l$ is not 
free in any $\alpha_i, \beta_j, \phi_k, \psi_{l'}$ for $l'\ne l$ and for all 
$i,j,k$. Thus, due to (B) and (D), we can rewrite (A) and (C) as:
(A'):  There exists  some $\overline{b} \in D_w$ and for all 
$\overline{d} \in D_w$ such that there exist $a_i \in D_w$ and  $v_i \in W$,
successor of $w$ such that\\ $M,v_i, \eta[ x_i \mapsto a_i;\ 
\overline{y} \mapsto \overline{b};\ \overline{z'} \mapsto \overline{d}] \vDash 
\alpha_i \land \bigwedge_j \beta_j \land \bigwedge_l \psi_l$.

(C'): there exists $\overline{b} \in D_w$ and for all $\overline{d} \in D_w$  
such that for all $c \in D_w$ there exist $v_1^c \dots v^c_{m_1} \in W$, 
successors of $w$ such that $M,v^c_k,\eta[ z_k\mapsto c;\ 
\overline{y} \mapsto \overline{b};\ \overline{z'} \mapsto \overline{d}] \models 
\phi_k \land \bigwedge_j \beta_j \land \bigwedge_l \psi_l$.

Now all the nodes in the conclusion of the $\BR$ rule have formulas as described 
in type A' or C' and are hence satisfiable.
\end{proof}

This proves the theorem, offering us a tableau construction procedure for
every formula: we have an open tableau iff the formula is satisfiable. Now note
that the tableau is not only of depth linear in the size of the formula, but 
also that subformulas are never repeated across siblings. Hence an algorithm 
can explore the tableau depthwise and reuse the same space when exploring other 
branches. The techniques are standard as in tableau procedures for modal logics. 
The extra space overhead for keeping track of domain elements is again only 
linear in the size of the formula. This way, we can get a \PSPACE-algorithm
for checking satisfiability. On the other hand, the $\PSPACE$ lower bound for
propositional modal logic applies as well, thus giving us the following corollary.

\begin{corollary}
Satisfiability of $\Bundle$-formulas is \PSPACE-complete. 
\end{corollary}

\subsection{Constant domain models}
We now take up the second task, to show that over constant domain models,
the culprit is the $\forall \Box$ bundle, by proving that the  satisfiability 
problem for the $\BundleEB$ is decidable over constant domain models.  
\cite{Wang17d} already showed decidability of the $\BundleEB$ over increasing
domain models. Taken together, we see that the $\BundleED$ fragment is
computationally robust.

The central idea behind the tableau procedure in the previous section was the 
use of existential quantifiers to offer their own witnesses, and cleanliness of
formulas ensures that these are new every time they are encountered. This works
well with increasing domain models, but in constant domain models, we need to fix
the domain right at the start of the tableau construction and use only these 
elements as witnesses. Yet, a moment's reflection assures us that we can give a
precise bound on how many new elements need to be added for each subformula of
the form $\existsBox{x}\phi$, and hence we can include as many elements as needed
at the beginning of the tableau construction.

Let $\Sub(\theta)$ stand for the finite set of subformulas of $\theta$.
Given a clean formula $\theta$ in NNF, for every 
$\exists x_j \Box \phi \in \Sub(\theta)$ let $\Var^\exists(\theta) = 
\{ x \mid \exists x \Box \phi \in \Sub(\theta)\}$. Now, cleanliness has
its advantages: every subformula of a clean formula is clean as well.
Hence, when $\theta_1$ and $\theta_2$ are both in $\Sub(\theta)$,
$\Var^\exists(\theta_1) \cap \Var^\exists(\theta_2) = \emptyset$. 
Similarly, when $\theta_1 \in \Sub(\theta)$ and $\theta_2 \in \Sub(\theta_1)$,
again $\Var^\exists(\theta_1) \cap \Var^\exists(\theta_2) = \emptyset$. 

Fix a clean formula $\theta$ in NNF with modal depth $h$.
For every $x \in \Var^\exists(\theta)$ define $\Var_x$ to be the set of $h$ 
\textit{fresh} variables $\{ x^k \mid 1\le k \le h\}$, and let $
\Var^+(\theta)=\bigcup\{\Var_x\mid x\in \Var^\exists(\theta)\}$, be the set 
of new variables to be added. Note that $\Var_x \cap \Var_y = \emptyset$ when
$x \neq y$. 
Fix a variable $z$ not occurring in $\theta$. Define $D_\theta= \FV(\theta) \cup 
\Var^+(\theta) \cup \{z \mid z$ does not occur in $\theta\}$. Note that $D_\theta$ is non-empty.

The tableau rules are given by:

 \begin{center}
 \noindent \begin{tabular}{|c|}
 \hline
 $\dfrac{w: \phi_1\lor\phi_2,\Gamma, \fixed}{w:\phi_1,\Gamma,\fixed \mid w:\phi_2,\Gamma,\fixed}$ \tt{($\lor$)}\qquad $\dfrac{w: \phi_1\land\phi_2,\Gamma,\fixed}{w:\phi_1,\phi_{2},\Gamma,\fixed}${\tt ($\land$)}\\
 \hline
 Given $n,s \geq 0;\ m\geq 1$:\\
 $\dfrac{w:\existsBox{x_1}\phi_1, \dots, \existsBox{x_n}\phi_n,\forallDiamond{y_1}\psi_1,\dots, \forallDiamond{y_m}\psi_m, r_1\dots r_s,\fixed}{ \langle (wv^y_{y_i}: \{\phi_j[x^{k_j}_j\slash x_j]\mid 1\leq j\leq n\}, \psi_i[y\slash y_i], \fixed') \rangle where\ y\in D_\theta, i\in [1,m]} $ \tt{($\BR$)}\\
\hline 

 Given $n\geq 1, s\geq 0$:\\
 $\dfrac {w:\exists x_1 \Box \phi_1, \cdots, \exists x_n \Box \phi_n, r_1, \cdots r_s, \fixed}{w: r_1\cdots r_s, C}  $ (\END) \\
 \hline
 \end{tabular}\\

 \noindent where $\fixed \subseteq D_\theta$ and $\fixed' = \fixed \cup \{ x^{k_j}_j\mid 1\leq j\leq n\} \cup \{ y\}$ where ${k_j}$ is the smallest number such that $x^{k_j}_j\in\Var_{x_j} \setminus \fixed$
 and $r_1 \dots r_s\in lit$. \end{center}

Note that the rule starts off one branch for each $y \in D_\theta$, since the $\forall 
\Diamond$ connective requires this over the fixed constant domain $D_\theta$. $\fixed$ 
keeps track of the variables used already along the path from the root till the current 
node.  These are now fixed, so the witness for $\exists x \Box \phi$ is picked from the 
remaining variables in $\Var_x(\theta)$.  Note that the variables in $\Var_{x_j}$ are
introduced only by applying $\BR$. Since $|\Var_{x_j}|$ is the modal depth, we always 
have a fresh $x_j^k$ to choose.

The notion of open tableau is as before, and the following observation is very useful:

\begin{proposition}\label{clean-BR}
The rule $(\BR)$ preserves cleanliness of formulas: if a tableau node is labelled
by $(w,\Gamma,C)$, $\Gamma$ is clean, and a child node labelled $(wv, \Gamma',C')$ is
created by $(\BR)$ then $\Gamma'$ is clean as well.
\end{proposition}

An important corollary of this proposition is that for all $x \in D_\theta$, at any
tableau node all occurrences of $x$ in $\Gamma$ are free. Therefore, for any
formula of the form $\psi_i[y/y_i]$ in the conclusion of the rule, $y$ is free and $y_i$ does not
occur at all in $\psi_i$. 

\begin{theorem}
 \label{thm: BundleEB constant decidability}
For any clean $\BundleEB$-formula $\theta$ in NNF, there is an open constant tableau 
from $(r, \{\theta\}, FV(\theta))$ iff $\theta$ is satisfiable in a constant domain model.
 \end{theorem}

 \begin{proof}

The structure of the proof is very similar to the earlier one, but we need to be
careful to check that sufficient witnesses exist as needed, since the domain is
fixed at the beginning of tableau construction. The proposition above, that the
rule $(\BR)$ preserves cleanliness of formulas, does the bulk of the work. The
details are presented in Appendix B.

\end{proof}

The complexity of the decision procedure does not change from before, since we
add only polynomially many new variables.
 
 \begin{corollary}
 \label{coro: BundleEB constant sat in PSPACE}
 The satisfiability problem for $\BundleEB$-formulas over constant domain models is 
\PSPACE-complete. 
 \end{corollary}

\section{Between Constant Domain and Increasing Domain}
We now show that the $\BundleEB$ fragment cannot distinguish increasing domain 
models and constant domain models. Note that in $\FOML$ this distinction is captured 
by the Barcan formula $\forall x \Box \phi\to \Box \forall x \phi$; but this 
is not expressible in  $\BundleEB$.\footnote{However, with equality added in the 
language we can distinguish the two by a formula. We can also accomplish this
in the $\forall\Box$ fragment: $\forall x \Box \forall y \Box \neg p(x) \land \forall z \Box \exists x \Diamond \neg p(x)$}

The tableau construction for the $\BundleEB$ fragment over increasing domain
models is a restriction of the one in the last section, and was presented in
\cite{Wang17d}.

 \begin{center}
  \noindent \begin{tabular}{|c|}
 \hline
 
 Given $n,s \geq 0;\ m\geq 1$:\\
 $\dfrac{w:\existsBox{x_1}\phi_1, \dots, \existsBox{x_n}\phi_n,\forallDiamond{y_1}\psi_1,\dots, \forallDiamond{y_m}\psi_m, r_1\dots r_s,F}{\langle (wv^y_{y_i}: \{\phi_j\mid 1\leq j\leq n\}, \psi_i[y\slash y_i], F') \rangle where\ y\in F', i\in [1,m]} $ \tt{($\BR$)}\\
 \hline 
 
 \end{tabular}\\
 \noindent where $F'=F\cup\{x_j\mid j\in [1,n] \}$.
 \end{center}


\begin{theorem}
For any $\BundleEB$ formula  $\phi$ satisfiable on some increasing domain model, 
the constant domain tableau of $\phi$ is open. 
\end{theorem}

\begin{proof}(Sketch)
We give a proof sketch. Consider a clean $\BundleEB$ formula $\phi$, and let
$\phi'=\phi\land \bigwedge \{\exists x' \Box \top \mid x' \in \Var^+(\phi)\}$ (recall that $\Var^+(\phi)=\bigcup_{x \in \Var^\exists(\phi)} \Var_x$). Clearly $\phi$ is satisfiable
in an increasing domain model iff $\phi'$ is as well. Let $T$ be an open tableau
for $\phi'$. We show that $T$ can be transformed into a constant open tableau $T'$
for $\phi$.

Suppose if $T$ has no applications of $(\BR)$, it is also a constant tableau and we are
done, so suppose that $T$ has at least one application of the rule $(\BR)$. By 
construction, all the $x'\in\Var^+(\phi)$ are added to the domain of the root, 
thus they are also at all the local domains in $T$. Note that we may have more 
elements in the local domains, such as $x$  that get added when we apply $\BR$ 
to $\exists x \Box \phi$, and therefore there are more branches than needed for 
a constant domain tableau of $\phi$ (such as those for $x$).

We can get rid of them by the following process:
\begin{itemize}
\item Fix $\psi=\exists x \Box \theta\in \Sub(\phi)$: 
\begin{itemize}
\item  Fix a node $s$ where  $\BR$ rule is applied and $\psi$ is in $s$. Since
$\phi$ is clean, there is no other node in any path of $T$ from the root passing 
through $s$ such that $\exists x \Box \theta'\in \Sub(\phi)$ occurs for some $\theta'$.
Let $m$ be the modal depth of $\phi$. The path from the root to the 
predecessor of $s$ can use at most $m-1$ different variables in  $\Var_x(\phi)$ 
when generating successors by applying the $\BR$ rule to some $\forall y\Diamond \theta$ 
formula.  Pick the first $x^h\in\Var_x$ which is not used in the path up to this node.  

\item Delete all the descendent nodes of $s$ that are named using $x^h$ when 
applying $\BR$ to some $\forall y \Diamond$ formula. It is not hard to see that 
the resulting sub-tableau rooted at $s$ has no occurrence of $x^h$ at all since 
$x^h$ could only be introduced among the children of $s$ using $\BR$. 

\item Rename all the occurrences of $x$ by $x^h$ (in formulas and node names) 
in all the descendent nodes of $s$. Then the branching structure from the 
sub-tableau rooted at $s$ will comply with the $\BR$ rule for constant-domain tableau. 

\item Repeat the above for all the application nodes of the $\BR$ rule w.r.t.\ $\psi$
\end{itemize}

\item Repeat the above procedure for all $\psi$ of the form $\exists x \Box \theta
\in \Sub(\phi)$.
\end{itemize}

The core idea is to simply use the newly introduced variable $x$ as if it were
$x^h$ in a constant-domain tableau.  Note that each branch-cutting operation and 
renaming operation (by new variables) above will preserve openness, since openness 
is merely about contradictions among literals. We then obtain a constant domain 
tableau by setting the domain as $D_\phi$.

\end{proof}

Note  that the constant domain tableau $T$ of $\phi$ constructed is a sub-tree 
embedding inside the increasing domain tableau $T'$ of $\phi'$.  However, showing
that it is generated precisely by the tableau rules in Section 4.2 involves some
tedious detail.

 \section{Discussion}
We have considered a decidable fragment of $\FOML$ over increasing domain models, by bundling quantifiers  together with modalities, and shown it to have the same complexity as propositional modal logic, while admitting arbitrary $k$-ary predicates. Considering that most decidable fragments of $\FOML$ involve severe syntactic restrictions involving quantifiers and variables, we have an interesting fragment for study. The tableau procedure offers a method of reasoning in the logic as well.

We also have a cautionary tale.  The $\exists \Box$ bundle, well motivated 
by considerations from epistemic logic (\cite{Wang17d}) is shown to be 
robustly decidable, for both constant domain and increasing domain semantics, 
whereas the $\forall \Box$ bundle is undecidable over constant domain models.

It should be emphasized that this paper is envisaged as a study of `bundling'
quantifiers and modalities (in terms of decidability) rather than proposing
{\sf the} definitive bundled fragment. The bundle $\Box \forall$ appears to
have properties similar to that of $\forall \Box$ (over constant domain
models) but $\Box \exists$ seems to be interestingly different. All the four combinations play important roles in Barcan formula, Buridan formula and their converses. Further,
it is not inconceivable that a bundle inspired by a particular shape of quantifier prefix such as $\exists x_1\dots \exists x_n \Box$ or $\exists x_1 \dots \exists x_n\forall z_1\dots\forall z_n\Box$ might be worthy of study, with their own motivation based on our knowledge about decidable prefix fragments of first-order logic. 

An obvious extension is to consider  the
language with constants, function symbols and equality. This leads to
not only familiar interesting conundrums regarding rigid identifiers
but computational considerations as well. Another direction is to find decidable bundled fragments of $\FOML$ over specific frame classes (such as T, S4, S5 etc.). It would be interesting to see how other non-normal modalities behave in the bundled fragments.

\bibliographystyle{aiml18}
\bibliography{ptml}

\medskip

\section*{Appendix A: Details of undecidability proof}
Here we present the proof details for undecidability of the $\BundleED$ over
constant domain models.

For any {\em quantifier free} $\FO(R)$ formula $\alpha$, we first recall the translation 
of $\alpha$ to $\BundleED$ formula $\phi_\alpha$ inductively as follows. 

\begin{itemize}
\item[-] $\phi_{R(x,y)} ::= \exists z \Diamond \big( p(x) \land q(y)\big)$, where $z$ is
distinct from $x$ and $y$. 
\item[-] $\phi_{\neg \alpha} ::= \neg \phi_{\alpha}$.
\item[-] $\phi_{\alpha_1 \land \alpha_2} ::= \phi_{\alpha_1} \land \phi_{\alpha_2}$.
\end{itemize}

Now consider an $\FO(R)$ sentence $\alpha$ (having no free variables) and presented in 
prenex form: $Q_1 x_1\ Q_2 x_2 \cdots Q_n x_n (\beta)$ where $\beta$ is quantifier free.
We define $\psi_\alpha$ to be the conjunction of the following three sentences:

\begin{itemize}
\item[-] $\psi_1 ::= Q_1 x_1 \Delta_1\  Q_2 x_2 \Delta_2\ \cdots Q_n x_n\Delta_n\  (\phi_{\beta})$ \\
where $Q_i x_i \Delta_i :=  \existsDiamond{x_i}$ if $Q_i = \exists$ and
$Q_i x_i \Delta_i := \forallBox{x_i}$ if $Q_i = \forall$.

\item[-] $\psi_2 ::= \forallBox{z_1}\forallBox{z_2}\big( (\existsDiamond{z})^{n}(\exists z\Diamond (p(z_1) \land q(z_2))) \Rightarrow (\forallBox{z})^n(\exists z\Diamond (p(z_1) \land q(z_2))) \big)$.

\item[-] $\psi_3 ::= \bigwedge_{j=1}^{n+2}(\forall\Box{z})^{j}\exists z\Diamond \top$.
\end{itemize}

\begin{theorem}
\label{thm: bundle-ED undecidability on constant domain}
An $FO(R)$ sentence $\alpha$ is satisfiable iff the $\BundleED$ sentence $\psi_\alpha$ is
constant domain satisfiable.
\end{theorem}

\begin{proof}
Fix $\alpha ::= Q_1 x_1 \cdots Q_n x_n \beta$, where $\beta$ is quantifier free. 
To prove $(\Rightarrow)$, assume that $\alpha$ is satisfiable. Let $D$ be some
domain such that $(D,I) \models \alpha$ where $I \subseteq (D \times D)$ is the 
interpretation for $R$.

Define $M = (W,R,D,\delta,\rho)$ where: 
\begin{itemize}
\item[ ]$W = \{ v_1, v_2 \} \cup \{w_i \mid 1 \le i \le n\} \cup \{ u_d \mid d \in D\}$.
\item[ ] $R =  \{ (v_1,v_2), (v_2,w_1)\} \cup \{ (w_i, w_{i+1}) \mid 1 \le i < n\} \cup \{ (w_n,u_d) \mid u_d \in W\}$.
\item[ ] $\delta(u) = D$ for all $u \in W$.
\item[ ] For all $i \in \{ 1,2\}$ and $1\le j \le n$ and $v_i,w_j \in W$ define $\rho(v_i,p) = \rho(v_i,q) =  \rho(w_j, p) = \rho(w_j,q) = \emptyset$ and for all $u_d \in W,\ \rho(u_d,p) = \{ d\}$ and $\rho	(u_d, q) = \{ c \mid (d,c) \in I \}$.
\end{itemize}

By construction, $M$ is a model that is a path of length $n+2$ originating from
$v_1$ until $w_n$ at which point we have a tree of depth $1$, with children $u_d$,
one for each $d \in D$. Therefore, it is easy to see that $M, v_1 \models \psi_3$.

Note that $M$ is a constant domain model. Further, it can be easily checked that
$(a,b) \in R$ iff $M, u_a \models (p(a) \land q(b))$. Thus $(D,I) \models R(x,y)$ iff
$M, w_n \models \exists z\Diamond (p(x) \land q(y))$. Hence a routine induction shows
that for any quantifier free formula $\beta'$, $(D,I) \models \beta'$ iff $M, w_n \models 
\phi_{\beta'}$. Further, since $M$ is a path model until $w_n$ and there is apath of length $n+3$ starting from $v_1$, we see that
$M, v_1 \models \psi_2 \land \psi_3$. We now claim that $M, v_1 \models \psi_1$, which would
complete the forward direction of the proof.

First, some notation. For all $1 \le i \le n$ let $x_1\cdots x_i$ be denoted by 
$\bar{x^i}$ and $\bar{d^i}$ be a vector of length $i$ with values in $D$. Let 
$[\bar{x^i} \rightarrow \bar{d^i}]$ denote the interpretation where $\sigma(x_j) 
= d_j$. Further, for $1 \leq i < n$, let $\alpha[i] = Q_{i+1} x_{i+1} \cdots 
Q_n x_n \beta$ and let $\psi_1[i] = Q_{i+1} \Delta _{i+1} \cdots Q_n x_n\Delta_n
\phi_\beta$.

The following claim proves that $M,v_1 \models \psi_1$: 

\paragraph{Claim.}  For all $1 \le i \le n$, $w_i \in W$, for all $d_1\cdots d_i \in D$, 
we have  $D,I,[\bar{x^i} \rightarrow \bar{d^i}] \models \alpha[i] $
iff $M,w_i,[\bar{x^i} \rightarrow \bar{d^i}] \models \psi_1[i]$.

The proof is by reverse induction on $i$. The base case, when $i = n$, follows from
our assertion above on the interpretation of $R$ at $w_n$.

Now for the induction step, we need to consider formulas $\alpha[i-1]$ and
$\psi_1[i-1])$ at $w_{i-1}$. Now $\alpha[i-1]$ is either $\exists {x_i} \alpha[i]$
or $\forall {x_i} \alpha[i]$. 

For the case when $\alpha[i-1]$ is $\exists x_{i} \alpha[i]$ the corresponding formula is 
$\existsDiamond{x_{i}}\psi_1[i]$.  We have \\
$D,I, [\bar{x^{i-1}} \rightarrow \bar{d^{i-1}}] \models \exists x_{i} \alpha[i]$ iff 
there is some $c\in D$ such that \\
$D,I,[\bar{x^{i-1}} \rightarrow \bar{d^{i-1}}, x_{i}\rightarrow c] \models \alpha[i]$  iff 
(by induction hypothesis) \\
$M,w_{i},[\bar{x^{i-1}}\rightarrow \bar{d^{i-1}}, x_{i}\rightarrow c]  \models \psi_1[i]$ iff  \\
$M,w_{i-1},[\bar{x^{i-1}}\rightarrow \bar{d^{i-1}}] \models \existsDiamond{x_{i}} \psi_1[i]$,
as required.

Now consider the case when $\alpha[i-1]$ is $\forall x_{i} \alpha[i]$, and let some $c \in D$.
Then $D,I,[\bar{x^{i-1}} \rightarrow \bar{d^{i-1}}, x_{i}\rightarrow c] \models \alpha[i]$ and
by induction hypothesis, \\
$M,w_{i},[\bar{x^{i-1}}\rightarrow \bar{d^{i-1}}, x_{i}\rightarrow c]  \models \psi_1[i]$.
Note that this holds for all $c \in D$. Additionally, $w_{i}$ is the unique successor
of $w_{i-1}$, and hence \\
$M,w_{i-1},[\bar{x^{i-1}}\rightarrow \bar{d^{i-1}}] \models \forall x_i\Box  \psi_1[i]$,
as required.

\paragraph{}
To prove $(\Leftarrow)$, suppose that $\psi_\alpha$ is satisfiable, and let  
$M= (W,D,R,\gamma,V)$ be a constant domain model such that $M,v \models \psi_\alpha$.
Without loss of generality, we can assume $(W,R)$ to be a tree rooted at $v$, and
$\psi_3$ ensures that every path in it has length at least $n+3$.  

Let $u'$ be any world at height $n+3$. Define $I_{u'} = \{ (c,d) \mid c \in \rho(u',p),
d \in \rho(u',q)\}$. For world $u$ at height $n+2$, define $I_u = \bigcup \{I_{u'} \mid
(u,u') \in R\}$. Since $M,v \models \psi_2$, we see that $I_u = I_w$, for all $u, w$
at height $n+2$. Hence we unambiguously define $I = I_u$, thus defining the first
order model $M' = (D,I)$. We now claim that the formula $\alpha$ is satisfied in
this model.

Recall that $\alpha = Q_1 x_1\ Q_2 x_2 \cdots Q_n x_n (\beta)$, where $\beta$ 
is quantifier-free.The definition of $I$ and the remark above ensure that 
$(D,I) \models \beta$ iff for all worlds $u$ at height $n+2$, 
$M,u \models \phi_{\beta}$. 

Let $w_i$ denote any world at height $i$, $3 \leq i \leq n+2$.

\paragraph{Claim.}  For all $i$, $3 \leq i \leq n+2$, for all 
$d_1\cdots d_i \in D$, we have: \\
$D,I,[\bar{x^i} \rightarrow \bar{d^i}] \models \alpha[i]$ iff 
for all $w_i \in W$ at height $i$, 
$M,w_i,[\bar{x^i} \rightarrow \bar{d^i}] \models \psi_1[i]$.

The proof of this claim is very similar to the proof in the forward direction
of the theorem, by induction on $n-i$. The base case is settled above and when
we consider the induction step, we strip one quantifier from $\alpha[i]$. 

Consider the case when $\alpha[i-1]$ is $\exists x_{i} \alpha[i]$; the 
corresponding formula is $\existsDiamond{x_{i}}\psi_1[i]$.  We have \\
$D,I, [\bar{x^{i-1}} \rightarrow \bar{d^{i-1}}] \models \exists x_{i} \alpha[i]$ iff
there is some $c\in D$ such that \\
$D,I,[\bar{x^{i-1}} \rightarrow \bar{d^{i-1}}, x_{i}\rightarrow c] \models \alpha[i]$  iff
(by induction hypothesis) \\
$M,w_{i},[\bar{x^{i-1}}\rightarrow \bar{d^{i-1}}, x_{i}\rightarrow c]  \models \psi_1[i]$ 
for every $w_i$ at height $i$ iff  \\
$M,w_{i-1},[\bar{x^{i-1}}\rightarrow \bar{d^{i-1}}] \models \existsDiamond{x_{i}} \psi_1[i]$,
for every $w_{i-1}$ at height $i-1$ as required.

Now suppose that $\alpha[i-1]$ is $\forall x_{i} \alpha[i]$, and let 
$D,I,[\bar{x^{i-1}} \rightarrow \bar{d^{i-1}} \models \forall x_{i} \alpha[i]$. 
Now let some $c \in D$.  Then \\
$D,I,[\bar{x^{i-1}} \rightarrow \bar{d^{i-1}}, x_{i}\rightarrow c] \models \alpha[i]$ and
by induction hypothesis, \\
$M,w_{i},[\bar{x^{i-1}}\rightarrow \bar{d^{i-1}}, x_{i}\rightarrow c]  \models \psi_1[i]$
for every $w_i$ at height $i$. But since the choice of $c$ was arbitrary,
this holds for all $c \in D$. Hence \\
$M,w_{i-1},[\bar{x^{i-1}}\rightarrow \bar{d^{i-1}}] \models \forall x_{i} \Box \psi_1[i]$,
for every $w_{i-1}$ at height $i-1$ as required. The other direction is similar.

\end{proof}

\section*{Appendix B: Details of constant tableau}
We now show that existence of a constant open tableau is equivalent to satisfiability
over constant domain models. We firstly recall a fact familiar from first order logic, 
that will be handy in the proof.

\begin{proposition}\label{prop.handy}
For any $\FOML$ formula $\phi$ and any model $M,w$: 
$$M, w, \sigma \vDash \phi[y\slash x]\iff  M, w, \sigma'[x\mapsto \sigma(y)]\vDash \phi$$
if $\sigma(y) = \sigma'(y)$ for all $y \neq x$ with $y$ not occurring free in $\phi$.
 \end{proposition}

\begin{theorem}
 \label{thm: BundleEB constant decidability}
 For any clean $\BundleEB$-formula $\theta$ in NNF, the following are equivalent: 
 \begin{itemize}
 \item There is an open constant tableau from $(r, \{\theta\}, FV(\theta))$.
 \item $\theta$ is satisfiable in a constant domain model.
 \end{itemize}
 \end{theorem}

 \begin{proof}

\noindent \textbf{Soundness of tableau construction}: \\
Given an open constant tableau $T = (W,V,E,\lambda)$ from the root node labelled
$(r,\{ \theta\},FV(\theta))$, we define $M=\{W, D_\theta, R, \rho\}$ where: 
 \begin{itemize}
 \item $w R v$ iff $v=wv'$ for some $v'$.
 \item $\overline{x}\in \rho(w, P)$ iff $P\overline{x} \in \Gamma$, where
$\lambda(t_w) = (w, \Gamma)$.
 \end{itemize}
 
By definition, $D_\theta$ is not empty. Further, $\rho$ is well-defined 
due to the openness of $T$. As before, we prove that $M,r$ is indeed a 
model of $\theta$, and this is proved by the following claim.

\paragraph{Claim.} For any tree node $w$ in $T$ if $\lambda(t_w) =
(w, \Gamma,C)$ and if $\alpha \in \Gamma$ then $(M,w,id_C) \models \alpha$.
(Again, we abuse notation and write $(M,w,C) \models \alpha$
for $(M,w,id_C) \models \alpha$.)

The proof proceeds by subtree induction on the structure of $T$. The base case
is when the node considered is a leaf node and hence it is also the last node
with that label. The definition of $\rho$ ensures that the literals are
evaluated correctly in the model. 

For the inductive step, the cases for application of conjunction
and disjunction rules are standard.

Consider the branching node $(w: \Gamma)$ where  
\[\Gamma=\{ \existsBox{x_1}\phi_1 \dots \existsBox{x_n} \phi_n, \forallDiamond{y_1}\psi_1 \dots \forallDiamond{y_m}\psi_m, r_1,\dots r_s\}. \]
By induction hypothesis, 
$$M, wv_{y_i}^y, C'(wv_{y_i}^y) \vDash \psi_i[y\slash y_i]\land \bigwedge_{1}^{n}\phi_j[x^{k_j}_j \slash x_j]$$ 
for every $y\in D_\theta$ and $i\in [1, m]$. We need to show that 
$M, w, C \vDash \phi$ for each $\phi \in \Gamma$. 

The assertion for literals in $\Gamma$ follows from the definition of $\rho$.
For each $\existsBox{x_j} \phi_j\in\Gamma$ and each $wv_{y_i}^y$, with $y \in D_\theta$,
we have $M, wv_{y_i}^y, C'(wv_{y_i}^y) \vDash \phi_j[x^{k_j}_j \slash x_j]$ by 
induction hypothesis. It is clear that $\{x^{k_j}_j\mid 1\leq j\leq n\}$ are not 
free in $\phi_j$ since they are chosen to be new. Hence, by Proposition \ref{prop.handy},
$M, wv_{y_i}^y, id_C[x_j\mapsto x^{k_j}_j] \vDash \phi_j$ for all $wv^y_{y_i}$. 
Therefore $M, w, id_C\vDash \existsBox{x_j} \phi_j$.

For $\forall y_i\Diamond \psi_i\in\Gamma$, and $y \in D_\theta$, by induction
hypothesis, we have $M,wv_{y_i}^y, C'(wv_{y_i}^y) \models \psi_i [y\slash y_i]$. 
By Proposition \ref{clean-BR} and its corollary, $y_i$ is not free in $\psi_i[ y\slash y_i]$ 
and hence  by Proposition \ref{prop.handy}, 
$M, wv^y_{y_i}, id_C[y_i\mapsto y]\vDash \psi_i$. Since this holds for each 
$y\in D_\theta$, we get $M, w, id_C \vDash \forall {y_i}\Diamond\psi_i$ for each $i$.

 Thus, it follows that $M,r,\sigma(r) \vDash \theta$. 

\medskip

\noindent \textbf{Completeness of tableau construction}: \\
We need to show that rule applications preserve the satisfiability of the formula set.
The proof is as before, we only discuss the $\BR$ case.

Consider a label set $\Gamma$ of clean formulas at a branching node. Let
$$\Gamma= \{\existsBox{x_j}\phi_j \mid j\in [1,n]\}  \cup  
\{\forallDiamond{y_i}\psi_i \mid j \in [1,m]\}\cup \{r_1\dots r_s\}$$ 
be satisfiable in a model $M=\{W, D, R, \rho\}$, $w \in W$ and an
assignment $\eta$ such that $M, w, \eta \vDash \phi$ for all $\phi \in \Gamma$.

By the semantics:
\begin{itemize}
\item[(A)] there exist $c_1, \dots c_{n} \in D$ such that for all $v \in W$, 
if $wRv$ then $M,v, \eta[x_j \mapsto c_j] \models \phi_j$.
 \item[(B)] for all $c\in D^M$ there exist $v_1^c \dots v^c_{m} \in W$, 
successors of $w$ such that $M,v^c_i,\eta[y_i \mapsto c] \models \psi_i$ 
for each $i\in[1, m]$.
 \end{itemize}

By cleanliness of formulas in $\Gamma$, each $x_j$ is free only in $\phi_j$, 
and each $y_i$ is free only in $\psi_i$. Thus we can merge the assignments 
without changing the truth values of $\phi_j$ and $\psi_i$, and obtain:
\begin{itemize}
\item[(B')] for all $c\in D$ there exist $v_1^c \dots v^c_{m} \in W$, successors of $w$,
such that $$M,v^c_i,\eta[\overline{x_j} \mapsto \overline{c_j}, y_i \mapsto c] \models \phi_1\land\dots \land \phi_n\land \psi_i, i\in[1, m].$$
 \end{itemize}

Fixing a $y\in D_\theta$ and an $i\in[1,m]$, in the following we show that 
$\{\phi_j[x^{k_j}_j\slash x_j]\mid 1\leq j\leq n\}, \psi_i[y\slash y_i]\}$ 
is satisfiable. There are two cases to be considered: 

\begin{itemize}

\item $y$ is not one of $x^{k_j}_j$. First since $\eta$ is an assignment for all 
the variables in $\Var$, we can suppose $\eta(y)=b\in D$. By $(B')$ above, there 
exists a successor $v^b_i$ of $w$ such that 
$$M,v^b_i,\eta[\overline{x_j} \mapsto \overline{c_j}, y_i \mapsto b] \models \phi_1\land\dots \land \phi_n\land \psi_i$$. Note that $\overline{x_j}$ and $y_i$ are not in 
$D_\theta$ thus they are different from $y$. On the other hand, by cleanliness of
$\Gamma$, $y_i$ does not occur in $\phi_j$ and $\eta(y)=b$, hence:
$$M,v^b_i,\eta[\overline{x_j} \mapsto \overline{c_j}] \models \phi_1\land\dots \land \phi_n\land \psi_i[y\slash y_i].$$ Finally, since each $x_j$ only occurs in $\phi_j$ and each 
$x^{k_j}_j$ does not occur in $\phi_1\dots \phi_j$ and $\psi_i[y\slash y_i]$, we have:
$$M,v^b_i,\eta[\overline{x^{k_j}_j} \mapsto \overline{c_j}] \models \phi_1[x^{k_1}_1\slash x_1]\land\dots \land \phi_n[x^{k_n}_n\slash x_n]\land \psi_i[y\slash y_i].$$

\item $y$ is $x^{k_j}_j$ for some $j$. Then we pick $c_j$, the witness for $x_j$, and by (B'), $$M,v^{c_j}_i,\eta[\overline{x_j} \mapsto \overline{c_j}, y_i \mapsto c_j] \models \phi_1\land\dots \land \phi_n\land \psi_i.$$ 
Since $y$ is $x^{k_j}_j$ then $$M,v^{c_j}_i,\eta[\overline{x_j} \mapsto \overline{c_j},x^{k_j}_j\mapsto c_j] \models \phi_1\land\dots \land \phi_n\land \psi_i[y\slash y_i].$$ 
Now proceeding similarly as in the case above we can show that:
$$M,v^{c_j}_i,\eta[\overline{x^{k_j}_j} \mapsto \overline{c_j}] \models \phi_1[x^{k_1}_1\slash x_1]\land\dots \land \phi_n[x^{k_n}_n\slash x_j]\land \psi_i[y\slash y_i].$$
 \end{itemize}
This completes the proof of the theorem.

  \end{proof}

\end{document}